\documentclass{llncs}

\usepackage{amsmath,amssymb}
\usepackage{enumitem}
\usepackage[dvips]{epsfig}
\usepackage{epsf}
\usepackage{float}
\usepackage{array}

\def\ovl{{\cal O}}
\def\olf{\ovl}
\def\tbar{{\overline{\bf t}}}

\def\fmod#1 #2{#1\ ({\rm mod}\ #2)}

\usepackage[dvips]{epsfig}
\usepackage{epsf}

\frontmatter

\title{Fife's Theorem Revisited}

\author{Jeffrey Shallit}

\institute{University of Waterloo,
Waterloo, ON  N2L 3G1 Canada 
\email{shallit@cs.uwaterloo.ca}
}

\begin{document}

\maketitle

\begin{abstract}
We give another proof of a theorem of Fife --- understood broadly
as providing a finite automaton that gives a complete description of all
infinite binary overlap-free words.  Our proof is significantly simpler
than those in the literature.  As an application we give a complete
characterization of the overlap-free words that are $2$-automatic.
\end{abstract}

\section{Introduction}

Repetitions in words is a well-researched topic.  Among the various themes
studied,
the binary overlap-free words play an important
role, both historically
and as an example exhibiting interesting structure.  Here by an
{\it overlap} we mean a word of the form $axaxa$, where $a$ is a single
letter and $x$ is a (possibly empty) word.

It is easy to see that neither
the finite nor the infinite binary overlap-free words
form a regular language.  Nevertheless,
in 1980, Earl Fife \cite{Fife:1980}
proved a theorem characterizing the infinite binary
overlap-free words as encodings of paths in a finite automaton.
His theorem was rather complicated to state and the proof was difficult.
Berstel \cite{Berstel:1994} later simplified the exposition, and
both Carpi \cite{Carpi:1993a} and Cassaigne \cite{Cassaigne:1993b} gave
an analogous analysis for the case of finite words.  Also see
\cite{Blondel&Cassaigne&Jungers:2009}.

In this note we show how to use the factorization
theorem of Restivo and Salemi \cite{Restivo&Salemi:1985a}
to give an alternate (and, we hope,
significantly simpler) proof
of Fife's theorem --- here understood in the general sense of providing
a finite automaton whose paths encode all infinite binary overlap-free words.

As a consequence we are able to disprove a conjecture on the fragility
of overlap-free words.

\section{Notation}

Let $\Sigma$ be a finite alphabet.  We let $\Sigma^*$ denote the set
of all finite words over $\Sigma$ and $\Sigma^\omega$ denote the
set of all (right-) infinite words over $\Sigma$.  We say
$y$ is a {\em factor} of a word $w$ if there exist words
$x, z$ such that $w = xyz$.

If $x$ is a finite word, then $x^\omega$ represents the
infinite word $xxx \cdots$.

As mentioned above,
an {\em overlap} is a word of the form $axaxa$, where $a \in \Sigma$
and $x \in \Sigma^*$.  An example of an overlap in English is the word
{\tt alfalfa}.  A finite or infinite
word is {\em overlap-free} if it contains no finite
factor that is an overlap. 

From now on we fix $\Sigma = \lbrace 0,1 \rbrace$.  The most famous
infinite binary overlap-free word is $\bf t$, the Thue-Morse word,
defined as the fixed point, starting with $0$, of the Thue-Morse
morphism $\mu$, which maps $0$ to $01$ and $1$ to $10$.  
We have
$$ {\bf t} = t_0 t_1 t_2 \cdots = 0110100110010110 \cdots .$$
The morphism $\mu$ has a second fixed point, $\tbar = \mu^\omega(1)$, which is
obtained from $\bf t$ by applying the complementation coding
defined by $\overline{0} = 1$ and $\overline{1} = 0$.

We let $\ovl$ denote the set of (right-) infinite binary overlap-free
words.

We now recall the infinite version of the
factorization theorem of Restivo and Salemi 
\cite{Restivo&Salemi:1985a} as stated in
\cite[Lemma 3]{Allouche&Currie&Shallit:1998}.

\begin{theorem}
Let ${\bf x} \in \ovl$, and let $P = \lbrace p_0, p_1, p_2, p_3, p_4 
\rbrace$, where $p_0 = \epsilon$, $p_1 = 0$, $p_2 = 00$,
$p_3 = 1$, and $p_4 = 11$.  Then there exists ${\bf y} \in \ovl$ and
$p \in P$ such that ${\bf x} = p \mu({\bf y})$.  Furthermore, this
factorization is unique, and $p$ is uniquely determined by inspecting
the first 5 letters of ${\bf x}$.
\end{theorem}

We can now iterate the factorization theorem to get

\begin{corollary}
Every infinite overlap-free word $\bf x$ can be written 
uniquely in the form
\begin{equation}
 {\bf x} = p_{i_1} \mu( p_{i_2} \mu ( p_{i_3} \mu ( \cdots ) ) ) \label{qq}
\end{equation}
with $i_j \in \lbrace 0, 1, 2, 3, 4 \rbrace$ for 
$j \geq 1$, subject to the understanding
that if there exists $c$ such that $i_j = 0$ for
$j \geq c$, then we also need to specify whether the ``tail'' of the
expansion represents $\mu^\omega(0) = {\bf t}$ or $\mu^\omega(1) = \tbar$.
Furthermore, every truncated expansion
$$p_{i_1} \mu(p_{i_2} \mu (p_{i_3} \mu (\cdots p_{i_{n-1}} \mu(p_{i_n})
\cdots )))$$
is a prefix of $\bf x$, with the understanding that if
$i_n = 0 $, then we need to replace $0$ with either
$1$ (if the ``tail'' represents $\bf t$) or $3$ (if the ``tail''
represents $\tbar$).
\end{corollary}

\begin{proof}
The form (\ref{qq}) is unique, since each $p_i$ is uniquely determined
by the first $5$ characters of the associated word.
\end{proof}

Thus, we can associate each infinite binary
overlap-free word $\bf x$ with the
essentially unique infinite sequence
of indices ${\bf i} := (i_j)_{j \geq 0}$ coding elements in $P$,
as specified by (\ref{qq}).  If $\bf i$ ends in $0^\omega$, then
we need an additional element (either $1$ or $3$) to disambiguate
between $\bf t$ and $\tbar$ as the ``tail''.  
In our notation, we separate this additional element with a
semicolon so that, for example, the string $000\cdots; 1$ represents
$\bf t$ and $000\cdots; 3$ represents $\tbar$.

Other sequences of interest include $203000\cdots; 1$, which codes
$ 001001 \tbar$, the lexicographically least infinite word, and
$2 (31)^\omega$, which codes the word having, in the $i$'th position,
the number of $0$'s in the binary expansion of $i$.

Of course, not every possible sequence of $(i_j)_{j \geq 1}$ of indices
corresponds to an
infinite overlap-free word.  For example, every infinite word coded
by $21 \cdots$ represents
$00 \mu( 0 \mu( \ldots))$ and hence
begins with $000$ and has an overlap.
Our goal is to characterize precisely,
using a finite automaton, those infinite sequences corresponding to
overlap-free words.

We recall some basic facts about overlap-free words.

\begin{lemma}
Let $a \in \Sigma$.  Then
\begin{itemize}
\item[(a)] ${\bf x} \in \olf \iff \mu({\bf x}) \in \olf$;
\item[(b)] $ a \, \mu({\bf x}) \in \olf \iff \overline{a} \, {\bf x} \in \olf$;
\item[(c)] $ a \, a \, \mu({\bf x}) \in \olf \iff \overline{a} 
\, {\bf x} \in \olf$
and $\bf x$ begins $\overline{a} \, a \, \overline{a}$.
\end{itemize}
\label{l1}
\end{lemma}

\begin{proof}
See, for example, \cite{Allouche&Currie&Shallit:1998}.
\end{proof}

We now define 11 subsets of $\ovl$:
\begin{eqnarray*}
A &=& \ovl \\
B &=& \lbrace {\bf x} \in \Sigma^\omega \ : \ 1 {\bf x} \in \ovl \rbrace \\
C &=& \lbrace {\bf x} \in \Sigma^\omega \ : \ 1 {\bf x} \in \ovl 
\text{ and } {\bf x} \text{ begins with } 101 \rbrace \\
D &=& \lbrace {\bf x} \in \Sigma^\omega \ : \ 0 {\bf x} \in \ovl \rbrace \\
E &=& \lbrace {\bf x} \in \Sigma^\omega \ : \ 0 {\bf x} \in \ovl 
\text{ and } {\bf x} \text{ begins with } 010 \rbrace  \\
F &=& \lbrace {\bf x} \in \Sigma^\omega \ : \ 0 {\bf x} \in \ovl 
\text{ and } {\bf x} \text{ begins with } 11 \rbrace \\
G &=& \lbrace {\bf x} \in \Sigma^\omega \ : \ 0 {\bf x} \in \ovl  
\text{ and } {\bf x} \text{ begins with } 1 \rbrace  \\
H &=& \lbrace {\bf x} \in \Sigma^\omega \ : \ 1 {\bf x} \in \ovl  
\text{ and } {\bf x} \text{ begins with } 1 \rbrace \\
I &=& \lbrace {\bf x} \in \Sigma^\omega \ : \ 1 {\bf x} \in \ovl 
\text{ and } {\bf x} \text{ begins with } 00 \rbrace \\
J &=& \lbrace {\bf x} \in \Sigma^\omega \ : \ 1 {\bf x} \in \ovl  
\text{ and } {\bf x} \text{ begins with } 0 \rbrace \\
K &=& \lbrace {\bf x} \in \Sigma^\omega \ : \ 0 {\bf x} \in \ovl  
\text{ and } {\bf x} \text{ begins with } 0 \rbrace \\
\end{eqnarray*}

Next, we describe the relationships between these classes:

\begin{lemma}
Let $\bf x$ be an infinite binary word.  Then
\begin{eqnarray}
{\bf x} \in A & \iff & \mu({\bf x}) \in A  \label{AA}\\
{\bf x} \in B & \iff & 0 \mu({\bf x}) \in A \label{AB} \\
{\bf x} \in C & \iff & 00 \mu({\bf x}) \in A \label{AC} \\
{\bf x} \in D & \iff & 1 \mu({\bf x}) \in A \label{AD} \\
{\bf x} \in E & \iff & 11 \mu({\bf x}) \in A \label{AE} \\
{\bf x} \in D & \iff & \mu({\bf x}) \in B \label{BD} \\
{\bf x} \in B & \iff & 0 \mu({\bf x}) \in B \label{BB} \\
{\bf x} \in E & \iff & 1 \mu({\bf x}) \in B \label{BE} \\
{\bf x} \in B & \iff & \mu({\bf x}) \in D \label{DB} \\
{\bf x} \in D & \iff & 1 \mu({\bf x}) \in D \label{DD} \\
{\bf x} \in C & \iff & 0 \mu({\bf x}) \in D \label{DC} \\
{\bf x} \in I & \iff & \mu({\bf x}) \in E \label{EI} \\
{\bf x} \in C & \iff & 0 \mu({\bf x}) \in E \label{EC} \\
{\bf x} \in F & \iff & \mu({\bf x}) \in C \label{CF} \\
{\bf x} \in E & \iff & 1 \mu({\bf x}) \in C \label{CE} \\
{\bf x} \in J & \iff & 0 \mu({\bf x}) \in I \label{IJ} \\
{\bf x} \in G & \iff & 1 \mu({\bf x}) \in F \label{FG} \\
{\bf x} \in K & \iff & \mu({\bf x}) \in J \label{JK} \\
{\bf x} \in J & \iff & \mu({\bf x}) \in K \label{KJ} 
\end{eqnarray}
\begin{eqnarray}
{\bf x} \in B & \iff & 0 \mu({\bf x}) \in J \label{JB} \\
{\bf x} \in C & \iff & 0 \mu({\bf x}) \in K \label{KC} \\
{\bf x} \in H & \iff & \mu({\bf x}) \in G \label{GH} \\
{\bf x} \in G & \iff & \mu({\bf x}) \in H \label{HG} \\
{\bf x} \in D & \iff & 1 \mu({\bf x}) \in G \label{GD} \\
{\bf x} \in E & \iff & 1 \mu({\bf x}) \in H \label{HE}
\end{eqnarray}
\label{AAK}
\end{lemma}

\begin{proof}

\ \\

(\ref{AA}):  Follows immediately from Lemma~\ref{l1} (a).
\medskip

(\ref{AB}), (\ref{AD}), (\ref{BD}), (\ref{DB}):  Follow
immediately from Lemma~\ref{l1} (b).
\medskip

(\ref{AC}), (\ref{AE}), (\ref{BE}), (\ref{DC}):  Follow
immediately from Lemma~\ref{l1} (c).
\medskip

(\ref{BB}):  $0 \mu({\bf x}) \in B
\iff 1 0 \mu({\bf x}) 
= \mu( 1 \, {\bf x}) \in \olf \iff 1 \, {\bf x} \in \olf$.
\medskip

(\ref{DD}):  Just like (\ref{BB}).
\medskip

(\ref{EI}):  $\mu({\bf x}) \in E \iff
(0 \mu({\bf x}) \in \olf$ and $\mu({\bf x})$ begins with $010)
\iff (1 {\bf x} \in \olf $ and ${\bf x}$ begins with $00$).
\medskip

(\ref{CF}):  Just like (\ref{EI}).
\medskip

(\ref{EC}):  $0 \mu({\bf x}) \in E \iff 
(00 \mu({\bf x}) \in \olf$ and $0 \mu({\bf x})$ begins with $010)
\iff (1 {\bf x} \in \olf$ and ${\bf x}$ begins with $101) $.
\medskip

(\ref{CE}):  Just like (\ref{EC}).
\medskip

(\ref{IJ}):  $0 \mu({\bf x}) \in I \iff
(10 \mu({\bf x}) \in \olf$ and $0 \mu({\bf x})$ begins with $00)
\iff (\mu(1 {\bf x}) \in \olf$ and ${\bf x}$ begins with $0)
\iff (1{\bf x} \in \olf$ and ${\bf x}$ begins with $0$).
\medskip

(\ref{FG}):  Just like (\ref{IJ}).
\medskip

(\ref{JK}):  $\mu({\bf x}) \in J \iff
(1 \mu({\bf x}) \in \olf$ and $\mu({\bf x})$ begins with $0) \iff
(0 {\bf x} \in \olf$ and ${\bf x}$ begins with $0)$.
\medskip

(\ref{GH}), (\ref{KJ}), (\ref{HG}):  Just like (\ref{JK}).
\medskip

(\ref{JB}):  $0 \mu({\bf x}) \in J \iff
(10 \mu({\bf x}) \in \olf$ and $0 \mu({\bf x})$
begins with $0) \iff \mu(1{\bf x}) \in \olf \iff 1{\bf x} \in \olf$.
\medskip

(\ref{GD}):  Just like (\ref{JB}).
\medskip

(\ref{KC}):  $0 \mu({\bf x}) \in K \iff
(0 0 \mu({\bf x}) \in \olf$ and $0 \mu({\bf x})$ begins with $0) \iff
(1 {\bf x} \in \olf$ and $\bf x$ begins with $101)$.
\medskip

(\ref{HE}):  Just like (\ref{KC}).

\end{proof}

We can now use the result of the previous lemma to create an
$11$-state
automaton that accepts all infinite sequences 
$(i_j)_{j \geq 1}$
over $\Delta := \lbrace 0,1,2,3, 4 \rbrace$ such that
$p_{i_1} \mu(p_{i_2} \mu(p_{i_3} \mu(\cdots )))$ is
overlap-free.  Each state represents one of the sets
$A, B, \ldots, K$ defined above, and the transitions are 
given by Lemma~\ref{AAK}.

Of course, we also need to verify that transitions not shown
correspond to the empty set of infinite words.
For example, a transition out of $B$ on the symbol $2$ would
correspond to the set
$\lbrace {\bf x} \ : 100 \mu({\bf x}) \in \olf \rbrace$.
But if $\bf x$ begins with $0$, then $100 \mu({\bf x}) = 10001 \cdots$
contains the overlap $000$ as a factor, whereas  if
$\bf x$ begins with $10$, then $100\mu({\bf x}) = 1001001 \cdots$
contains the overlap $1001001$ as a factor, and if
$\bf x$ begins with $11$, then $100\mu({\bf x}) = 1001010 \cdots$
contains $01010$ as a factor.  Similarly, we can (somewhat tediously)
verify that all other transitions not given in Figure~\ref{nice} 
correspond to the empty set:

\begin{eqnarray*}
\delta(B, 4) &=& \lbrace
{\bf x} \in \Sigma^\omega \ : \ 111 \mu({\bf x}) \in {\cal O} 
\rbrace = \emptyset\\
\delta(D, 2) &=& \lbrace {\bf x} \in \Sigma^\omega \ : \ 000 \mu({\bf x}) \in {\cal O} \rbrace 
= \emptyset\\
\delta(D, 4) &=& \lbrace {\bf x}  \in \Sigma^\omega\ :\  011 \mu({\bf x}) \in {\cal O} \rbrace
= \emptyset \\
\delta(C, 1) &=& \lbrace {\bf x} \in \Sigma^\omega \ : \ 10 \mu({\bf x}) \in {\cal O} \text{ and }
0 \mu({\bf x}) \text{ begins with } 101 \rbrace = \emptyset\\
\delta(C, 2) &=& \lbrace {\bf x} \in \Sigma^\omega \ : \ 100 \mu({\bf x}) \in {\cal O} \text{ and }
00 \mu({\bf x}) \text{ begins with } 101 \rbrace = \emptyset\\
\delta(C, 4) &=& \lbrace {\bf x} \in \Sigma^\omega \ : \ 111 \mu({\bf x}) \in {\cal O} 
\text{ and } 11 \mu({\bf x})  \text{ begins with } 101 \rbrace = \emptyset\\
\delta(E,2) &=& \lbrace {\bf x} \in \Sigma^\omega \ : \ 000 \mu({\bf x}) \in {\cal O} 
\text{ and } 00\mu({\bf x}) \text{ begins with } 010 \rbrace = \emptyset\\
\delta(E, 3) &=& \lbrace {\bf x} \in \Sigma^\omega \ : \ 01\mu({\bf x}) \in {\cal O} 
\text{ and } 1\mu({\bf x}) \text{ begins with } 010 \rbrace = \emptyset\\
\delta(E, 4) &=& \lbrace {\bf x} \in \Sigma^\omega \ : \ 011\mu({\bf x}) \in {\cal O} 
\text{ and } 11\mu({\bf x}) \text{ begins with } 010 \rbrace = \emptyset\\
\delta(F, 0) &=& \lbrace {\bf x} \in \Sigma^\omega \ : \ 0 \mu({\bf x}) \in {\cal O}
\text{ and } \mu({\bf x}) \text{ begins with } 11 \rbrace = \emptyset\\
\delta(F, 1) &=& \lbrace {\bf x} \in \Sigma^\omega \ : \ 00 \mu({\bf x}) \in {\cal O}
\text{ and } 0\mu({\bf x}) \text{ begins with } 11 \rbrace = \emptyset\\
\delta(F, 2) &=& \lbrace {\bf x} \in \Sigma^\omega \ : \ 000\mu({\bf x}) \in {\cal O} 
\text{ and } 00\mu({\bf x}) \text{ begins with } 11 \rbrace = \emptyset\\
\delta(F, 4) &=& \lbrace {\bf x} \in \Sigma^\omega \ :  \ 011\mu({\bf x}) \in {\cal O}
\text{ and } 11\mu({\bf x}) \text{ begins with } 11 \rbrace = \emptyset\\
\delta(J, 2) &=& \lbrace {\bf x} \in \Sigma^\omega \ : \ 100 \mu({\bf x}) \in {\cal O}
\text{ and } 00\mu({\bf x}) \text{ begins with } 0 \rbrace = \emptyset\\
\delta(J, 3) &=& \lbrace {\bf x} \in \Sigma^\omega \ : \ 11 \mu({\bf x}) 
\in {\cal O} \text{ and } 1\mu({\bf x}) \text{ begins with } 0 \rbrace = \emptyset\\
\delta(J, 4) &=& \lbrace {\bf x} \in \Sigma^\omega \ : \ 111\mu({\bf x}) 
\in {\cal O} \text{ and } 11\mu({\bf x}) \text{ begins with } 0 \rbrace = \emptyset\\
\delta(K,2) &=& \lbrace {\bf x} \in \Sigma^\omega \ : \ 000 \mu({\bf x}) \in {\cal O}
\text{ and } 00\mu({\bf x}) \text{ begins with } 0 \rbrace = \emptyset\\
\delta(K,3) &=& \lbrace {\bf x} \in \Sigma^\omega \ : \ 01 \mu({\bf x}) \in {\cal O} 
\text{ and } 1\mu({\bf x}) \text{ begins with } 0 \rbrace = \emptyset\\
\delta(K, 4) &=& \lbrace {\bf x} \in \Sigma^\omega \ : \  011\mu({\bf x} \in {\cal O}) 
\text{ and } 11\mu({\bf x}) \text{ begins with } 0 \rbrace = \emptyset\\
\end{eqnarray*}

The proof of most of these is immediate.  (We have not listed
$\delta(I,a)$ for $a \in \lbrace 0,2,3,4 \rbrace$, nor 
$\delta(G,a)$ for $a \in \lbrace 1,2,4 \rbrace$, nor
$\delta(H,a)$ for $a \in \lbrace 1,2,4 \rbrace$, as these are symmetric
with other cases.)  The only one that requires some thought is
$\delta(F,4)$:
\begin{itemize}
\item If $\bf x$ begins $00$, then $011\mu({\bf x}) = 0110101 \cdots$, which
has $10101$ as a factor.
\item If $\bf x$ begins $01$, then $011\mu({\bf x}) = 0110110 \cdots$, which
has $0110110$ as a factor.
\item If $\bf x$ begins $1$, then $011\mu({\bf x}) = 01110 \cdots$, which
has $111$ as a factor.
\end{itemize}

\begin{figure}[H]
\begin{center}
	\epsfig{file=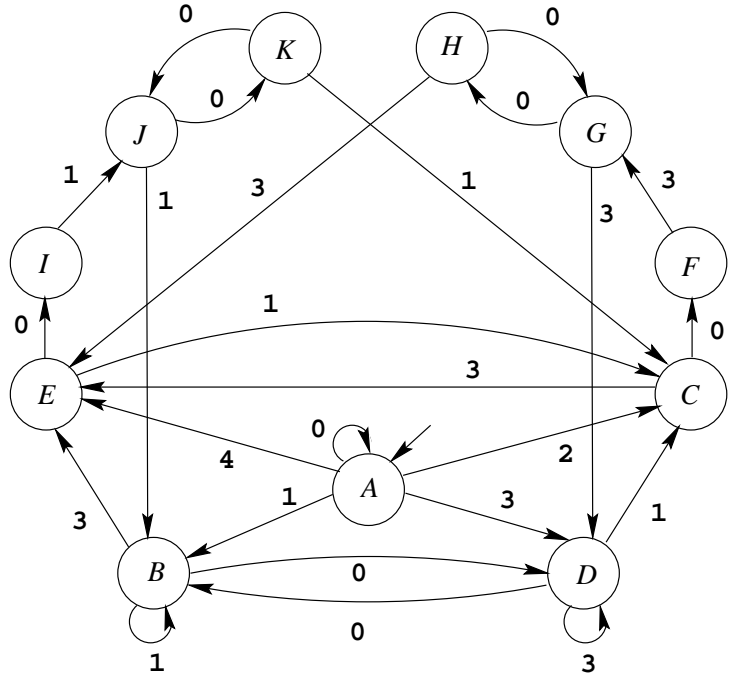}
\end{center}
\caption{Automaton coding infinite binary overlap-free words}
\label{nice}
\end{figure}

From Lemma~\ref{AAK} and the results above, we get

\begin{theorem}
    Every infinite binary overlap-free word $\bf x$ is encoded by an
infinite path, starting in $A$, through the automaton in Figure~\ref{nice}.

   Every infinite path through the automaton
not ending in $0^\omega$ codes a unique
infinite binary overlap-free word $\bf x$.  If a path $\bf i$ ends in
$0^\omega$ and this suffix corresponds to a cycle on state A or a cycle
between states B and D, then $\bf x$ is 
coded by either ${\bf i}; 1$ or
${\bf i}; 3$.  If a path $\bf i$ ends in $0^\omega$
and this suffix corresponds to a cycle between states J and K,
then $\bf x$ is coded by ${\bf i}; 1$.  If a path $\bf i$ ends
in $0^\omega$ and this suffix corresponds to a cycle between
states G and H, then $\bf x$ is coded by ${\bf i}; 3$.
\label{main}
\end{theorem}

\begin{corollary}
Each of the 11 sets $A, B, \ldots, K$ is uncountable.
\end{corollary}

\begin{proof}
We prove this for $K$, with the proof for the other sets being similar.
Elements in the set $K$ correspond to those infinite paths leaving the state $K$
in Figure~\ref{nice}.  It therefore suffices to produce uncountably
many distinct paths leaving $K$.  One way to do this, for example, is
by $\lbrace 13010, 1301000 \rbrace^\omega$.
\end{proof}

\section{The lexicographically least overlap-free word}

We now recover a theorem of \cite{Allouche&Currie&Shallit:1998}:

\begin{theorem}
The lexicographically least infinite binary overlap-free word  is
$001001 \tbar$.
\end{theorem}

\begin{proof}
Let $\bf x$ be the lexicographically least infinite word, and 
let $\bf y$ be its code.  Then ${\bf y}[1]$ must be $2$, since any
other choice codes a word that starts with $01$ or something lexicographically
greater.  Once ${\bf y}[1] = 2$ is chosen, the next two symbols must be
${\bf y}[2..3] = 03$.   Now we are in state $G$.  We argue that the lexicographically
least string that follows causes us to alternate
between states $G$ and $H$ on $0$, producing $100\cdots$.
For otherwise our only choices are $30$, $31$,
or (if we are in $G$) $33$ as the next two symbols, and all of these code
a word lexicographically greater than $100$.
Hence ${\bf y} = 203\, 0^\omega; 1$ is the code for the
lexicographically least sequence, and this codes $001001 \tbar$.
\end{proof}

\section{Automatic infinite binary overlap-free words}

As a consequence of Theorem~\ref{main}, we can give a complete description
of the infinite binary overlap-free words that are $2$-automatic
\cite{Allouche&Shallit:2003}.  Recall that an infinite word $(a_n)_{n \geq 0}$
is $k$-automatic if there exists a deterministic finite automaton with
output that, on input $n$ expressed in base $k$, produces an output
associated with the state last visited that is equal to $a_n$.

\begin{theorem}
An infinite binary overlap-free word is $2$-automatic if and only
if its code is both
specified by the DFA given above in Figure 1, and is ultimately periodic.
\label{auto}
\end{theorem}

First, we need two lemmas:

\begin{lemma}
An infinite binary word ${\bf x} = a_0 a_1 a_2 \cdots$
is $2$-automatic if and only if
$\mu({\bf x})$ is $2$-automatic.
\label{tm}
\end{lemma}

\begin{proof}
For one direction, we use the fact that the class of $k$-automatic
sequences is closed under uniform morphisms 
(\cite[Theorem 6.8.3]{Allouche&Shallit:2003}).  So if $\bf x$ is
$2$-automatic, so is $\mu({\bf x})$.

For the other, we use the well-known characterization of automatic
sequences in terms of the $k$-kernel
\cite[Theorem 6.6.2]{Allouche&Shallit:2003}:  a sequence
$(c_n)_{n \geq 0}$ is $k$-automatic if and only if its $k$-kernel
defined by
$$ \lbrace (c_{k^e n + i})_{n \geq 0} \ : \ e \geq 0 \text{ and }
	0 \leq i < k^e \rbrace$$
is finite.  Furthermore, each sequence in the $k$-kernel is $k$-automatic.

Now if ${\bf y} = \mu({\bf x}) = b_0 b_1 b_2 \cdots$, 
then $b_{2n} = a_n$.  So one of the sequences in the $2$-kernel of
$\bf y$ is $\bf x$, and if $\bf y$ is $2$-automatic, then so is $\bf x$.
\end{proof}

Now we can prove Theorem~\ref{auto}.

\begin{proof}
Suppose the code of $\bf x$ is ultimately periodic.  Then we can write
its code as $y z^\omega$ for some finite words $y$ and $z$.  
Since the class of $2$-automatic sequences is closed under appending
a finite prefix \cite[Corollary 6.8.5]{Allouche&Shallit:2003}, by
Lemma~\ref{tm}, it suffices to show that the word coded by $z^\omega$
is $2$-automatic.

The word $z^\omega$ codes an overlap-free word ${\bf w}$ satisfying
${\bf w} = t \varphi ({\bf w})$, where $t$ is a finite word and
$\varphi$ is a power of $\mu$.  If $t$ is empty the result is clear.
Otherwise, by iteration, we get that
\begin{equation}
{\bf w} = t \varphi(t) \varphi^2(t) \cdots .
\label{varph}
\end{equation}

The $2$-kernel of a sequence is obtained by repeated {\it $2$-decimation},
that is, recursively splitting a sequence into its even- and odd-indexed
terms.  When we apply $2$-decimation to $\mu^k(t)$, where $t$ is a finite word,
we get $\mu^{k-1} (t)$ and $\mu^{k-1} ( \, \overline{t} \, )$.  These words are both
of even length, provided $k$ is at least $1$.  Hence iteratively applying $2$-decimation
to $\bf w$, as given in (\ref{varph}), shows that if $\varphi = \mu^k$, then
the $2$-kernel of $\bf w$ is contained in
$$S := \lbrace u \mu^i(v) \mu^{i+k}(v) \mu^{i+2k}(v) \cdots \ : \ |u| \leq 2|t|
\text{ and } v \in \lbrace t, \overline{t} \rbrace \text{ and } 1 \leq i \leq k \rbrace,$$
which is a finite set.

On the other hand, suppose the code for
$\bf x$ is not ultimately periodic.  Then we show that the
$2$-kernel is infinite.    To see this, note that the code for $\bf x$
contains a $2$ or $4$ only at the beginning, so we can assume without
loss of generality that the code for $\bf x$ contains only the
letters $0, 1, 3$.  Now it is easy to see that if the code for $\bf x$
is $a{\bf y}$ for some letter $a \in \lbrace 0, 1, 3 \rbrace$ and
infinite string ${\bf y} \in \lbrace 0, 1, 3 \rbrace^\omega$, then
one of the sequences in the $2$-kernel (obtained by taking either the odd-
or even-indexed terms) is either 
coded by $\bf y$ or its complement is coded by $\bf y$.
Since the code for $\bf x$ is not ultimately periodic,
there are infinitely many distinct sequences in the orbit of the
code for $\bf x$, under the shift.  (By the orbit of $\bf y$ we mean the set of
sequences of the form ${\bf y}[i..\infty]$ for $i \geq 1$.)
Now infinitely many of these sequences correspond to a sequence in the $2$-kernel, or its
complement.  Hence $\bf x$ is not $2$-automatic.
\end{proof}

\section{A fragility conjecture disproved}

Brown, Rampersad, Shallit, and Vasiga showed that the Thue-Morse
word $\bf t$ is {\it fragile} in the following sense:  if any finite
nonempty set of positions is chosen, and the bits in those positions are
simultaneously flipped to the complement of their original values, the result
has an overlap \cite{Brown&Rampersad&Shallit&Vasiga:2006}.

It is natural to wonder if a similar result holds more generally for
all overlap-free words.  However, the statement must be modified in this
more general setting, as (for example) both $0 {\bf t}$ and $1 {\bf t}$ are
overlap-free.

The author made the following conjecture at the
Oberwolfach meeting in 2010:

\begin{conjecture}
For each infinite binary overlap-free word $\bf w$ there exists a constant
$C$ (depending on $\bf w$) such that if the bits at
any finite nonempty set of positions
$> C$ are flipped, then the result has an overlap.
\end{conjecture}

Using our result we can disprove this conjecture.  For consider
the infinite words coded by $1 \lbrace 113011,313011 \rbrace^\omega$.
By examining the automaton,
each such word is easily seen to be a valid
code for an overlap-free word.
These words have blocks that line up exactly at the same positions, but 
each $6th$ block can be replaced by the appropriate power of
$\mu$ evaluated at either $0$ or $1$, and each such choice gives a 
distinct overlap-free word.  

\section{Remarks}

According to a theorem of Karhum\"aki and the author 
\cite{Karhumaki&Shallit:2004}, there is a similar factorization
theorem for all exponents $\alpha$ with $2 < \alpha \leq {7 \over 3}$.
Recently we have proven similar results for $\alpha = {7 \over 3}$
\cite{Rampersad&Shallit&Shur:2011}.

I am grateful to the referees for a careful reading of the manuscript.

\end{document}